\documentclass[english]{article}
\usepackage{amsfonts,amsmath,amsthm,amscd,amssymb,latexsym}

\theoremstyle{plain}
\newtheorem{thm}{Theorem}[section]

\newtheorem{prop}[thm]{Proposition}

\theoremstyle{definition}
\theoremstyle{remark}
\numberwithin{equation}{section}
\newcommand{\keywords}{\textbf{Key words and phrases: }\medskip}
\newcommand{\subjclass}{\textbf{Math. Subj. Clas.: }\medskip}

\begin{document}
\title{\textbf{A discrete model of the Dirac-K\"{a}hler equation} }
\author{\textbf{Volodymyr Sushch} \\
{ Koszalin University of Technology} \\
 { Sniadeckich 2, 75-453 Koszalin, Poland} \\
 {volodymyr.sushch@tu.koszalin.pl} }

\date{}
\maketitle
\begin{abstract}
We construct a new discrete analog  of the Dirac-K\"{a}hler equation in which some key geometric aspects of the continuum counterpart are captured.
We describe a discrete Dirac-K\"{a}hler equation in the intrinsic notation as a set of difference equations and prove several statements about its decomposition into difference equations of Duffin type. We study an analog of gauge transformations for the massless discrete Dirac-K\"{a}hler equations.
\end{abstract}


 
 \section{Introduction}
The problem of finding discrete analogies of various continuous models is of considerable interest in physical theories.
In this work we study a discrete counterpart of the Dirac-K\"{a}hler equation which describes  fermion fields in terms of inhomogeneous differential forms. It is well-known that a very useful way to discretize quantum field equations is provided by a lattice approach. However, there are some difficulties in capturing fundamental properties of differential operators on the lattice, for example, the differential defined as a difference operator in the lattice formulation does not satisfy the Leibniz rule. We propose a discretization scheme based on the use of the differential form language in which the exterior derivative  $d$, the Hodge star operator $\ast$ and the exterior product $\wedge$ of differential forms are replaced by their discrete analogies. The algebraic relations between these operators are captured in the proposed discrete model.
We are interested in a discrete Dirac-K\"{a}hler equation only from the mathematics point of view.
Our approach bases on the differential geometric formalism proposed by Dezin \cite{Dezin}. We adapt a combinatorial model of Minkowski space from \cite{S3} and define discrete analogs of the operators $d, \ast$ and $\wedge$ in the same way as in \cite{S2}. Alternative geometric discretisation schemes which use also a discrete exterior calculus has already been developed by several authors  \cite{Becher,  Dodziuk, Rabin, SSSA, Teixeira, Wilson}. In the lattice formulation the first attempts to construct a discrete model  of the Dirac-K\"{a}hler equation based on differential geometric methods  were done by Rabin \cite{Rabin} and by Becher and Joos \cite{Becher}. In a sequel of papers \cite{Adda, Beauce,  Campos, Catterall1, Catterall2, Catterall, Kanamori} the Dirac-K\"{a}hler equation on the lattice has been extensively studied.
Our approach is close to one proposed by Rabin \cite{Rabin}. The obtained discrete Dirac-K\"{a}hler and Duffin equations are formally the same as in  \cite{ Rabin}. However, our discrete analog of the Hodge star operator is different than that given in \cite{Rabin}. In \cite{Rabin}, this operation is defined by ussing dual lattice.  We define a star operation by using a discrete analog of the exterior product. This improvement allows us to preserve the Lorentz metric structure in our discrete model and introduce an inner product of discrete forms (cochains). We may therefore construct a discrete codifferential so as to be the formal adjoint of a coboundary operator. Consequently, we obtain more precise difference equations.

The main purpose of this paper is to construct a discrete model of the Dirac-K\"{a}hler equation which preserves the geometric structure of the continuum counterpart.
We show that some key properties of the Dirac-K\"{a}hler system that hold in the smooth setting also hold in the discrete case. It is known \cite{Cantoni} that the Dirac-K\"{a}hler equation  decomposes uniquely into four uncoupled equations of Duffin type. We prove the same in the discrete case.
Recently  much attention has been directed to the study of the massless Dirac-K\"{a}hler field \cite{Kruglov1, Malik, Pletyuhov}. From the physics point of view the Dirac-K\"{a}hler system has three massless limits \cite{Kruglov1}. In \cite{Malik, Pletyuhov}, it has shown that  the Kalb-Ramond field (notoph) equations and the electromagnetic equations are partial cases of the Dirac-K\"{a}hler massless equation. We formulate a discrete version of these massless cases. It is well known that the electromagnetic field is invariant under the gauge transformations (see for instance \cite{Nakahara}). We study analogies of these transformations for the discrete model.

\section{Smooth settings}
Let $M={\mathbb R}^{1,3}$ be  Minkowski space with  metric signature  $(+,-,-,-)$.
Denote by $\Lambda^r(M)$ the vector space of smooth differential $r$-form, $r=0,1,2,3,4$.
 Let $\overset{r}{\omega}$ and  $\overset{r}{\varphi}$ be  $r$-forms on $M$. The inner product is defined by
\begin{equation}\label{2.1}
(\overset{r}{\omega}, \ \overset{r}{\varphi})=\int_{M}\overset{r}{\omega}\wedge\ast\overset{r}{\varphi},
\end{equation}
where $\wedge$ is the exterior product and $\ast$ is the Hodge star operator  $\ast:\Lambda^r(M)\rightarrow\Lambda^{4-r}(M)$ (with respect to the Lorentz metric).
Let $d:\Lambda^r(M)\rightarrow\Lambda^{r+1}(M)$ be the exterior differential and let $\delta:\Lambda^r(M)\rightarrow\Lambda^{r-1}(M)$ be the formal adjoint of $d$  with respect to (\ref{2.1}) (the codifferential). We have $\delta=-\ast d\ast$.  Then the Laplacian (Laplace-Beltrami operator) acting on $r$-forms is defined by
\begin{equation}\label{2.2}
\Delta\equiv -d\delta-\delta d:\Lambda^r(M)\rightarrow\Lambda^{r}(M).
\end{equation}
It is clear that $-d\delta-\delta d=(d-\delta)^2$.
For  Minkowski space the operator $\Delta$ is the d'Alembert (wave) operator.
 Then for any $r$-form $\overset{r}{\omega}\in\Lambda^r(M)$  the Klein-Gordon  equation can be written as
 \begin{equation}\label{2.3}
 \Delta\overset{r}{\omega}=m^2\overset{r}{\omega},
\end{equation}
where $m$  is a mass parameter.
Denote by $\Lambda(M)$ the set of all differential forms on $M$. We have
$\Lambda(M)=\Lambda^0(M)\oplus\Lambda^1(M)\oplus\Lambda^2(M)\oplus\Lambda^3(M)\oplus\Lambda^4(M)$. Let $\Omega\in\Lambda(M)$
be an inhomogeneous differential form.   This form can be expanded  as
\begin{equation*}
\Omega=\sum_{r=0}^4\overset{r}{\omega},
\end{equation*}
where $\overset{r}{\omega}\in\Lambda^r(M)$.  The Dirac-K\"{a}hler equation is given by
\begin{equation}\label{2.4}
(d-\delta)\Omega=m\Omega.
\end{equation}
First proposed by K\"{a}hler \cite{Kahler}, this  equation is the generalization of the Dirac equation.
It is easy to show, that Equation (\ref{2.4}) is equivalent to the following equations
\begin{align*}\label{}
-\delta\overset{1}{\omega}=m\overset{0}{\omega},\\
d\overset{0}{\omega}-\delta\overset{2}{\omega}=m\overset{1}{\omega},\\
d\overset{1}{\omega}-\delta\overset{3}{\omega}=m\overset{2}{\omega},\\
d\overset{2}{\omega}-\delta\overset{4}{\omega}=m\overset{3}{\omega},\\
d\overset{3}{\omega}=m\overset{4}{\omega}.
\end{align*}
There are three massless limits of this system \cite{Pletyuhov}. Let us consider $m=\sqrt{m_1m_2}$. Then we obtain the massless Dirac-K\"{a}hler equation in the case if any of mass parameters $m_1, m_2$ (or both simultaneously) is equal to zero.

\section{Combinatorial model of Minkowski space}

Following \cite{Dezin}, let the tensor product  $C(4)=C\otimes C\otimes C\otimes C$
of a   1-dimensional complex be a combinatorial model of Euclidean space
 ${\mathbb R}^4$. The 1-dimensional complex $C$ is defined in the following way.
Let $C^0$ denotes the real linear space of 0-dimensional chains generated by
basis elements $x_\kappa$ (points), $\kappa\in {\mathbb Z}$. It is convenient to
introduce the shift operators  $\tau,\sigma$ in the set of indices by
\begin{equation}\label{3.1}
\tau\kappa=\kappa+1, \qquad \sigma\kappa=\kappa-1.
\end{equation}
We denote the open interval $(x_\kappa, x_{\tau\kappa})$ by $e_\kappa$.
One can regard the set $\{e_{\kappa}\}$ as a set of basis elements of
the real linear space $C^1$. Suppose that $C^1$ is the space of 1-dimensional
chains.
Then the 1-dimensional complex (combinatorial real line) is the direct sum
of the introduced spaces $C=C^0\oplus C^1$. The boundary operator $\partial$
in $C$ is given by
$$
\partial x_\kappa=0, \qquad  \partial e_\kappa=x_{\tau\kappa}-x_\kappa.
$$
The definition is extended to arbitrary chains by linearity.

Multiplying the basis elements $x_\kappa, e_\kappa$ in various way we obtain
basis elements of $C(4)$. Let $s_k$ be  an arbitrary basis element of $C(4)$. Then we have
${s_k=s_{k_0}\otimes s_{k_1}\otimes s_{k_2}\otimes s_{k_3}}$, where $s_{k_i}$
is either  $x_{k_i}$ or  $e_{k_i}$ and $k_i\in {\mathbb Z}$. Here $k=(k_0, k_1, k_2, k_3)$ is a multi-index.
The 1-dimensional basis elements
of $C(4)$ can be written as
\begin{align}\label{3.2}
e_k^0=e_{k_0}\otimes x_{k_1}\otimes x_{k_2}\otimes x_{k_3}, \nonumber \qquad
e_k^1=x_{k_0}\otimes e_{k_1}\otimes x_{k_2}\otimes x_{k_3}, \\
e_k^2=x_{k_0}\otimes x_{k_1}\otimes e_{k_2}\otimes x_{k_3},  \qquad
e_k^3=x_{k_0}\otimes x_{k_1}\otimes x_{k_2}\otimes e_{k_3},
\end{align}
where  the superscript $i$ indicates  a place of $e_{k_i}$ in $e_k^i$ and  $i=0,1,2,3$.
In the same way we will write  the 2-dimensional basis elements
of $C(4)$ as
\begin{align}\label{3.3}
e_k^{01}=e_{k_0}\otimes e_{k_1}\otimes x_{k_2}\otimes x_{k_3},  \qquad
e_k^{12}=x_{k_0}\otimes e_{k_1}\otimes e_{k_2}\otimes x_{k_3},\nonumber \\
e_k^{02}=e_{k_0}\otimes x_{k_1}\otimes e_{k_2}\otimes x_{k_3}, \qquad
e_k^{13}=x_{k_0}\otimes e_{k_1}\otimes x_{k_2}\otimes e_{k_3}, \nonumber \\
e_k^{03}=e_{k_0}\otimes x_{k_1}\otimes x_{k_2}\otimes e_{k_3}, \qquad
e_k^{23}=x_{k_0}\otimes x_{k_1}\otimes e_{k_2}\otimes e_{k_3}.
\end{align}
Denote by $e_k^{012}, e_k^{013}, e_k^{023}, e_k^{123}$ the 3-dimensional basis elements
of $C(4)$.

Let $C(4)=C(p)\otimes C(q)$, where $p+q=4$. If $a\in C(p)$ and $b\in C(q)$ are arbitrary chains, belonging to the complexes being multiplied, then  we extend the definition of   $\partial$ to chains of $C(4)$ by the rule
\begin{equation}\label{3.4}
\partial(a\otimes b)=\partial a\otimes b+(-1)^ra\otimes\partial b,
\end{equation}
where $r$ is the dimension of the chain $a$.

Suppose that the combinatorial model of Minkowski space has the same structure
as $C(4)$.  We will use the index $k_0$ to denote the basis elements of $C$ which correspond to the time coordinate of
$M$. Hence the indicated basis elements will be written as  $x_{k_0}$, $e_{k_0}$.

Let us now consider a dual complex to $C(4)$. We define its as the complex of cochains
$K(4)$ with real coefficients. The complex $K(4)$
has a similar structure, namely ${K(4)=K\otimes K\otimes K\otimes K}$, where $K$ is a dual
complex to the 1-dimensional complex $C$. We will write the basis elements of $K$  as
$x^\kappa$ and $e^\kappa$, $\kappa\in {\mathbb Z}$. Then an arbitrary basis element of $K(4)$ can be written  as
${s^k=s^{k_0}\otimes s^{k_1}\otimes s^{k_2}\otimes s^{k_3}}$, where $s^{k_i}$
is either  $x^{k_i}$ or  $e^{k_i}$. Denote by $s_{(r)}^k$  an $r$-dimensional basis element of $K(4)$. Here the symbol $(r)$ contains the whole requisite information about the number and situation of $e^{k_i}\in K$  in $s_{(r)}^k$. For example, the  1-dimensional basic elements
$e_i^k\in K(4)$ and the  2-dimensional basic elements
$e_{ij}^k\in K(4)$ have the form (\ref{3.2}) and (\ref{3.3})  respectively. We will call cochains forms,
emphasizing their relationship with the corresponding continual objects, differential forms.

Denote by $K^r(4)$ the set of forms of degree $r$. We can represent  $K(4)$ as
\begin{equation*}
K(4)=K^0(4)\oplus K^1(4)\oplus K^2(4)\oplus K^3(4)\oplus K^4(4).
\end{equation*}
Let $\overset{r}{\omega}\in K^r(4)$. Then we have
\begin{equation}\label{3.5}
\overset{0}{\omega}=\sum_k\overset{0}{\omega}_kx^k, \quad  \mbox{where} \quad x^k=x^{k_0}\otimes x^{k_1}\otimes x^{k_2}\otimes x^{k_3},
\end{equation}
\begin{equation}\label{3.6}
\overset{1}{\omega}=\sum_k\sum_{i=0}^3\omega_k^ie_i^k, \qquad
\overset{2}{\omega}=\sum_k\sum_{i<j} \omega_k^{ij}e_{ij}^k, \qquad
\overset{3}{\omega}=\sum_k\sum_{i<j<l} \omega_k^{ijl}e_{ijl}^k,
\end{equation}
\begin{equation}\label{3.7}
\overset{4}{\omega}=\sum_k\overset{4}{\omega}_ke^k,    \quad  \mbox{where} \quad  e^k=e^{k_0}\otimes e^{k_1}\otimes e^{k_2}\otimes e^{k_3}.
\end{equation}
Here the components $\overset{0}{\omega}_k, \ \omega_k^i, \ \omega_k^{ij}, \ \omega_k^{ijl}$ and $\overset{4}{\omega}_k$ are real numbers.

As in \cite{Dezin}, we define the pairing (chain-cochain) operation for any basis elements
$\varepsilon_k\in C(4)$,  $s^k\in K(4)$ by the rule
\begin{equation}\label{3.8}
\langle\varepsilon_k, \ s^k\rangle=\left\{\begin{array}{l}0, \quad \varepsilon_k\ne s_k\\
                            1, \quad \varepsilon_k=s_k.
                            \end{array}\right.
\end{equation}
The operation (\ref{3.8}) is linearly extended to arbitrary chains and cochains.

The coboundary operator $d^c: K^r(4)\rightarrow K^{r+1}(4)$ is defined by
\begin{equation}\label{3.9}
\langle\partial a, \ \overset{r}{\omega}\rangle=\langle a, \ d^c\overset{r}{\omega}\rangle,
\end{equation}
where $a\in C(4)$ is an $r+1$ dimensional chain. The operator $d^c$ is an analog of the exterior differential.
From the above it follows that
\begin{equation*}\label{}
 d^c\overset{4}{\omega}=0 \quad \mbox{and} \quad d^cd^c\overset{r}{\omega}=0 \quad \mbox{for any} \quad r.
\end{equation*}
Let us introduce for convenient  the shifts
operator $\tau_i$ and $\sigma_i$ as
\begin{equation}\label{3.10}\tau_ik=(k_0,...\tau
 k_i,...k_3), \quad
 \sigma_ik=(k_0,...\sigma k_i,...k_3), \quad i=0,1,2,3,
  \end{equation}
  where $\tau$ and $\sigma$ are defined by (\ref{3.1}).
Using (\ref{3.8}) and (\ref{3.9}) we can calculate
\begin{equation}\label{3.11}
d^c\overset{0}{\omega}=\sum_k\sum_{i=0}^3(\Delta_i\overset{0}{\omega}_k)e_i^k,
\end{equation}
\begin{equation}\label{3.12}
d^c\overset{1}{\omega}=\sum_k\sum_{i<j}(\Delta_i\omega_k^j-\Delta_j\omega_k^i)e_{ij}^k,
\end{equation}
\begin{align}\label{3.13}
d^c\overset{2}{\omega}=\sum_k\big[(\Delta_0\omega_k^{12}-\Delta_1\omega_k^{02}+\Delta_2\omega_k^{01})e_{012}^k\nonumber \\
+(\Delta_0\omega_k^{13}-\Delta_1\omega_k^{03}+\Delta_3\omega_k^{01})e_{013}^k \nonumber \\
+(\Delta_0\omega_k^{23}-\Delta_2\omega_k^{03}+\Delta_3\omega_k^{02})e_{023}^k \nonumber \\
+(\Delta_1\omega_k^{23}-\Delta_2\omega_k^{13}+\Delta_3\omega_k^{12})e_{123}^k\big],
\end{align}
\begin{equation}\label{3.14}
d^c\overset{3}{\omega}=\sum_k(\Delta_0\omega_k^{123}-\Delta_1\omega_k^{023}+\Delta_2\omega_k^{013}-\Delta_3\omega_k^{012})e^k,
\end{equation}
where $\Delta_i$ is the difference operator defined by
\begin{equation}\label{3.15}
\Delta_i\omega_k=\omega_{\tau_ik}-\omega_k
\end{equation}
for any components $\omega_k$ of $\overset{r}{\omega}$. For simplicity of notation we write here $\omega_k$ instead of $\omega_k^{(r)}$.

Let us now introduce in $K(4)$ a multiplication which is an analog of the
exterior multiplication for differential forms.
Denote by  $K(r)$ the $r$-dimensional complex, ${r=1,2,3}$. We define the $\cup$-multiplication by induction on $r$.
Suppose that the $\cup$-multiplication in $K(r)$ has been defined. Then we introduce it for basis elements of
$K(r+1)$ by the rule
\begin{equation}\label{3.16}
(s^k_{(p)}\otimes s^\kappa)\cup(s^k_{(q)}\otimes s^\mu)=
Q(\kappa,q)(s^k_{(p)}\cup s^k_{(q)})\otimes(s^\kappa\cup s^\mu),
\end{equation}
where $s^k_{(p)}, s^k_{(q)}\in K(r)$, $s^\kappa(s^\mu)$ is either $x^\kappa(x^\mu)$
or $e^\kappa(e^\mu)$, $\kappa, \mu\in{\mathbb Z}$, and the signum function
$Q(\kappa, q)$ is equal to $-1$ if the dimension of both elements $s^\kappa$,
$s_{(q)}^k$ is odd and to $+1$ otherwise.
For the basis elements of $K(1)=K$ the $\cup$-multiplication is defined as follows
\begin{equation*}\label{}
x^\kappa\cup x^\kappa=x^\kappa, \quad e^\kappa\cup x^{\tau\kappa}=e^\kappa,
\quad x^\kappa\cup e^\kappa=e^\kappa, \quad \kappa\in{\mathbb Z},
\end{equation*}
supposing the product to be zero in all other case. To arbitrary forms the
$\cup$-multiplication can be extended linearly.

\begin{prop}
Let $\varphi$ and $\psi$ be arbitrary forms of $K(4)$.
Then
\begin{equation}\label{3.17}
 d^c(\varphi\cup\psi)=d^c\varphi\cup\psi+(-1)^r\varphi\cup
d^c\psi,
\end{equation}
where  $r$ is the degree  of a form $\varphi$.
\end{prop}
The proof can be found in \cite[p.147]{Dezin}. Note that Relation (\ref{3.17}) is a discrete version of the Leibniz rule for differential forms.

By definition, the coboundary operator $d^c$ and the
$\cup$-multiplication do not depend on a metric. Hence they have the same structure in
$K(4)$ as in the case of the combinatorial Euclidean space.
 At the same time, to define a discrete analog of the Hodge star operator
$\ast$ we must take into account the  Lorentz metric structure  on $K(4)$.
Define the operation $\ast: K^r(4)\rightarrow K^{4-r}(4)$ for an arbitrary basis element $s^k=s^{k_0}\otimes  s^{k_1}\otimes s^{k_2}\otimes s^{k_3}$ by the rule
\begin{equation}\label{3.18}
 s^k\cup\ast s^k=Q(k_0) e^{k},
\end{equation}
where $Q(k_0)$ is equal to $+1$ if $s^{k_0}=x^{k_0}$ and
to $-1$ if  $s^{k_0}= e^{k_0}$.
 For example, for the 1-dimensional basis elements
 $e_i^k$  we have $e_0^k\cup\ast e_0^k=-e^k$ and  $e_i^k\cup\ast e_i^k=e^k$ \ for $i=1,2,3$. Recall that $e^k=e^{k_0}\otimes e^{k_1}\otimes e^{k_2}\otimes e^{k_3}$ is the 4-dimensional basic element of $K(4)$. Relation (\ref{3.18})  preserves the  Lorentz  signature of metric in our  discrete model. From (\ref{3.18}) we obtain
\begin{align}\label{3.19, 3.20}
\ast x^k=\ast(x^{k_0}\otimes x^{k_1}\otimes x^{k_2}\otimes x^{k_3})=e^k, \\
\ast e^k=-x^{\tau k_0}\otimes x^{\tau k_1}\otimes x^{\tau k_2}\otimes x^{\tau k_3}=-x^{\tau k},
\end{align}
\begin{equation}\label{3.21}
\ast e_0^k=-e_{123}^{\tau_0 k}, \qquad \ast e_1^k=-e_{023}^{\tau_1 k}, \qquad
\ast e_2^k=e_{013}^{\tau_2 k}, \qquad \ast e_3^k=-e_{012}^{\tau_3 k},
\end{equation}
\begin{align}\label{3.22}
\ast e_{01}^k&=-e_{23}^{\tau_{01} k}, \qquad \ast e_{02}^k=e_{13}^{\tau_{02} k}, \qquad \ast e_{03}^k=-e_{12}^{\tau_{03} k}, \nonumber \\
\ast e_{12}^k&=e_{03}^{\tau_{12} k}, \qquad  \ast e_{13}^k=-e_{02}^{\tau_{13} k}, \qquad \ast e_{23}^k=e_{01}^{\tau_{23} k},
\end{align}
\begin{equation}\label{3.23}
\ast e_{012}^k=-e_3^{\tau_{012} k}, \quad \ast e_{013}^k=e_2^{\tau_{013} k}, \quad
\ast e_{023}^k=-e_1^{\tau_{023} k}, \quad \ast e_{123}^k=-e_0^{\tau_{123} k}.
\end{equation}
Here we use $\tau_{ij}$ and $\tau_{ijl}$ to denote the operators which shift to the right the indicated components of $k=(k_0,k_1,k_2,k_3)$, for example,
\begin{equation*}
\tau_{12}k=(k_0,\tau k_1,\tau k_2,k_3), \quad
 \tau_{023}k=(\tau k_0,k_1,\tau k_2,\tau k_3),
  \end{equation*}
  and $\tau k=(\tau k_0,\tau k_1,\tau k_2,\tau k_3)$.
  It is easy to check that
 \begin{equation*}
\ast\ast s^k_{(r)}=(-1)^{r+1}s^{\tau k}_{(r)},
 \end{equation*}
 where $s^k_{(r)}$ is an $r$-dimensional basic element of $K(4)$.
 Then if we perform the $\ast$ operation twice on any $r$-form $\overset{r}{\omega}\in K(4)$,  we  obtain
 \begin{equation*}
\ast\ast\overset{r}{\omega}=\ast\ast\sum_k\sum_{(r)}\omega_k^{(r)} s^k_{(r)}=(-1)^{r+1}\sum_k\sum_{(r)}\omega_k^{(r)}s^{\tau k}_{(r)}=
(-1)^{r+1}\sum_k\sum_{(r)}\omega_{\sigma k}^{(r)}s^k_{(r)},
 \end{equation*}
where   $\sigma k=(\sigma k_0,\sigma k_1,\sigma k_2,\sigma k_3)$.

Let $V\subset C(4)$ be some fixed "domain" of the complex $C(4)$.
We can written $V$ as follows
\begin{equation}\label{3.24}
V=\sum_ke_k, \qquad k=(k_0,k_1,k_2,k_3),\qquad k_i=1,2, ...,N_i,
\end{equation}
where $e_k=e_{k_0}\otimes e_{k_1}\otimes e_{k_2}\otimes e_{k_3}$
is the 4-dimensional basis element of $C(4)$. We agree that
in what follows the subscripts $k_i,\  i=0,1,2,3$, always run the set
of values indicated in (\ref{3.24}).
Suppose that the  forms (\ref{3.5})--(\ref{3.7}) are vanished on $C(4)\setminus V$, i.e., if $k_i<1$ or $k_i>N_i$ then  $\omega_k^{(r)}=0$ for any $r$-form $\overset{r}{\omega}\in K(4)$.
 For forms
$\overset{r}{\varphi}, \ \overset{r}{\omega}\in K^r(4)$ of the same degree $r$ the inner product is defined by the relation
\begin{equation}\label{3.25}
(\overset{r}{\varphi}, \ \overset{r}{\omega})_V=\langle V, \ \overset{r}{\varphi}\cup\ast\overset{r}{\omega}\rangle.
\end{equation}
For the forms of different degrees the product (\ref{3.25}) is set equal to zero. See also \cite{S1}.
The definition imitates correctly the continual case (\ref{2.1}). Using (\ref{3.8}) and (\ref{3.19, 3.20})--(\ref{3.23})
 we obtain
\begin{align*}\label{}
(\overset{0}{\omega}, \ \overset{0}{\omega})_V&=\sum_k(\overset{0}{\omega}_k)^2, \\
(\overset{1}{\omega}, \ \overset{1}{\omega})_V&=\sum_k\big[-(\omega_k^0)^2+(\omega_k^1)^2+(\omega_k^2)^2+(\omega_k^3)^2\big], \\
(\overset{2}{\omega}, \ \overset{2}{\omega})_V&=\sum_k\big[-(\omega_k^{01})^2-(\omega_k^{02})^2-(\omega_k^{03})^2+(\omega_k^{12})^2+(\omega_k^{13})^2+
(\omega_k^{23})^2\big], \\
(\overset{3}{\omega}, \ \overset{3}{\omega})_V&=\sum_k\big[-(\omega_k^{012})^2-(\omega_k^{013})^2-(\omega_k^{023})^2+(\omega_k^{123})^2\big], \\
(\overset{4}{\omega}, \ \overset{4}{\omega})_V&=-\sum_k(\overset{4}{\omega}_k)^2.
\end{align*}
\begin{prop}
 Let $\varphi\in K^r(4)$  and $\psi\in K^{r+1}(4)$. Then we have
\begin{equation}\label{3.26}
 (d^c\varphi, \ \psi)_V=(\varphi, \ \delta^c\psi)_V,
\end{equation}
 where
 \begin{equation}\label{3.27}
 \delta^c\psi=(-1)^{r+1}\ast^{-1}d^c\ast\psi
 \end{equation}
  is the
operator formally adjoint of $d^c$.
\end{prop}
\begin{proof} From (\ref{3.9}), (\ref{3.17}) and (\ref{3.25}) we obtain
\begin{align*} (d^c\varphi,\
\psi)_V&=\langle V, \ d^c\varphi\cup\ast\psi\rangle=\langle V, \
 (d^c(\varphi\cup\ast\psi)-(-1)^r \varphi\cup d^c\ast\psi)\rangle\\
&=\langle\partial V, \ \varphi\cup\ast\psi\rangle+(-1)^{r+1}
\langle V, \ \varphi\cup\ast(\ast^{-1}d^c\ast\psi)\rangle\\
&=\langle\partial V,\
\varphi\cup\ast\psi\rangle+(-1)^{r+1}(\varphi, \ \ast^{-1}d^c\ast\psi)_V,
\end{align*}
where we used $\ast\ast^{-1}=1$.
It remains to prove that ${\langle\partial V, \ \varphi\cup\ast\psi\rangle=0}$.
Using (\ref{3.4}) we can calculate
\begin{align*}
\partial e_k&=\partial (e_{k_0}\otimes e_{k_1}\otimes e_{k_2}\otimes e_{k_3})\\
&=e_{\tau_0k}^{123}-e_k^{123}-e_{\tau_1k}^{023}+e_k^{023}+e_{\tau_2k}^{013}-e_k^{013}-e_{\tau_3k}^{012}+e_k^{012}.
\end{align*} From this we obtain
\begin{align*}
\partial V=\sum_k \big(e_{\tau N_0k_1k_2k_3}^{123}-
e_{1k_1k_2k_3}^{123}-e_{k_1\tau N_1k_2k_3}^{023}+
e_{k_01k_2k_3}^{023}\\
+e_{k_0k_1\tau N_2k_3}^{013}-
e_{k_0k_11k_3}^{013}-e_{k_0k_1k_2\tau N_3}^{012}+
e_{k_0k_1k_21}^{012}\big).
\end{align*}
Then if we compute the 3-form $\varphi\cup\ast\psi$ on $\partial V$, we obtain the expression   which consists of only the terms
\begin{equation*}
\varphi_{k_0...\tau N_i...k_3}\psi_{k_0...N_i...k_3}, \qquad
\varphi_{k_0...1...k_3}\psi_{k_0...0...k_3},
\end{equation*}
where $\varphi_{k_0k_1k_2k_3}=\varphi_k$ and  $\psi_{k_0k_1k_2k_3}=\psi_k$ are components of the forms $\varphi$  and  $\psi$.
Since by assumption, we have $\varphi_{k_0...\tau N_i...k_3}=\psi_{k_0...0...k_3}=0$  for all $i=0,1,2,3$, it follows that
 ${\langle\partial V, \ \varphi\cup\ast\psi\rangle=0}$.
\end{proof}
The operator $\delta^c: K^{r+1}(4) \rightarrow K^r(4)$ is a discrete analog of the codifferential $\delta$. For the 0-form (\ref{3.5}) we have $\delta^c\overset{0}{\omega}=0$.
 Using (\ref{3.11})--(\ref{3.14}) and (\ref{3.26}) we can calculate
\begin{equation}\label{3.28}
\delta^c\overset{1}{\omega}=\sum_k(-\Delta_0\omega_{\sigma_0k}^{0}+\Delta_1\omega_{\sigma_1k}^{1}+\Delta_2\omega_{\sigma_2k}^{2}+\Delta_3\omega_{\sigma_3k}^{3})x^k,
\end{equation}
\begin{align}\label{3.29} \nonumber
\delta^c\overset{2}{\omega}=\sum_k\big[(\Delta_1\omega_{\sigma_1k}^{01}+\Delta_2\omega_{\sigma_2k}^{02}+\Delta_3\omega_{\sigma_3k}^{03})e_{0}^k\\ \nonumber
+(\Delta_0\omega_{\sigma_0k}^{01}+\Delta_2\omega_{\sigma_2k}^{12}+\Delta_3\omega_{\sigma_3k}^{13})e_{1}^k\\ \nonumber
+(\Delta_0\omega_{\sigma_0k}^{02}-\Delta_1\omega_{\sigma_1k}^{12}+\Delta_3\omega_{\sigma_3k}^{23})e_{2}^k\\
+(\Delta_0\omega_{\sigma_0k}^{03}-\Delta_1\omega_{\sigma_1k}^{13}-\Delta_2\omega_{\sigma_2k}^{23})e_{3}^k\big],
\end{align}
\begin{align}\label{3.30} \nonumber
\delta^c\overset{3}{\omega}=\sum_k\big[(\Delta_2\omega_{\sigma_2k}^{012}+\Delta_3\omega_{\sigma_3k}^{013})e_{01}^k+
(-\Delta_1\omega_{\sigma_1k}^{012}+\Delta_3\omega_{\sigma_3k}^{023})e_{02}^k\\ \nonumber
+(-\Delta_1\omega_{\sigma_1k}^{013}-\Delta_2\omega_{\sigma_2k}^{023})e_{03}^k
+(-\Delta_0\omega_{\sigma_0k}^{012}+\Delta_3\omega_{\sigma_3k}^{123})e_{12}^k\\
+(-\Delta_0\omega_{\sigma_0k}^{013}-\Delta_2\omega_{\sigma_2k}^{123})e_{13}^k
+(-\Delta_0\omega_{\sigma_0k}^{023}+\Delta_1\omega_{\sigma_1k}^{123})e_{23}^k\big],
\end{align}
\begin{align}\label{3.31}
\delta^c\overset{4}{\omega}=\sum_k\big[-(\Delta_3\omega_{\sigma_3k}^{4})e_{012}^k+(\Delta_2\omega_{\sigma_2k}^{4})e_{013}^k\\ \nonumber
-(\Delta_1\omega_{\sigma_1k}^{4})e_{023}^k-(\Delta_0\omega_{\sigma_0k}^{4})e_{123}^k\big].
\end{align}
It is obvious that
$\delta^c\delta^c\overset{r}{\omega}=0$ \ for any $r=1,2,3,4.$

A discrete analog of the Laplace-Beltrami operator  (\ref{2.2}) is defined by
\begin{equation}\label{3.32}
\Delta^c=-(d^c\delta^c+\delta^cd^c): \ K^r(4) \rightarrow K^r(4).
\end{equation}
We have
\begin{equation}\label{3.33}
-(d^c\delta^c+\delta^cd^c)=(d^c-\delta^c)^2.
\end{equation}
\begin{prop}
 The operator $\Delta^c$ is self adjoint for all forms supported in $V$, i.e.,
\begin{equation}\label{3.34}
 (\Delta^c\varphi, \ \psi)_V=(\varphi, \ \Delta^c\psi)_V.
\end{equation}
 \end{prop}
\begin{proof}
By (\ref{3.26}) it is obvious.
 \end{proof}
 Due to the Lorentz signature of the metric of $K(4)$ the operator $\Delta^c$ is a discrete analog of the d'Alembert (wave) operator.
 Then for any $r$-form $\overset{r}{\omega}$ a discrete analog of the Klein-Gordon  equation (\ref{2.3}) can be written as
 \begin{equation}\label{3.35}
 \Delta^c\overset{r}{\omega}=m^2\overset{r}{\omega},
\end{equation}
where $m$  is a mass parameter.
 \section{Discrete Dirac-K\"{a}hler equation}
Define a discrete inhomogeneous form as follows
\begin{equation}\label{4.1}
\Omega=\sum_{r=0}^4\overset{r}{\omega},
\end{equation}
where $\overset{r}{\omega}$ is given by (\ref{3.5})--(\ref{3.7}).
Let us consider the  discrete Klein-Gordon  equation on these inhomogeneous forms:
\begin{equation}\label{4.2}
 \Delta^c\Omega=m^2\Omega.
\end{equation}
It is clear that Equation (\ref{4.2}) is equivalent to five equations (\ref{3.35}) for $r=0,1,2,3,4$. Using (\ref{3.33}) it is possible to  introduce a discrete analog of the Dirac-K\"{a}hler equation (\ref{2.4}) by the rule
\begin{equation}\label{4.3}
(d^c-\delta^c)\Omega=m\Omega.
\end{equation}
We can write this equation more explicitly by separating its homogeneous components as
\begin{align}\label{4.4} \nonumber
-\delta^c\overset{1}{\omega}=m\overset{0}{\omega},\\ \nonumber
d^c\overset{0}{\omega}-\delta^c\overset{2}{\omega}=m\overset{1}{\omega},\\
d^c\overset{1}{\omega}-\delta^c\overset{3}{\omega}=m\overset{2}{\omega},\\ \nonumber
d^c\overset{2}{\omega}-\delta^c\overset{4}{\omega}=m\overset{3}{\omega},\\ \nonumber
d^c\overset{3}{\omega}=m\overset{4}{\omega}.
\end{align}
Using (\ref{3.11})--(\ref{3.14}) and (\ref{3.28})--(\ref{3.31})  equations (\ref{4.4}) can be written as the set of following difference equations
\begin{align*}\label{}
\Delta_0\omega_{\sigma_0k}^{0}-\Delta_1\omega_{\sigma_1k}^{1}-\Delta_2\omega_{\sigma_2k}^{2}-\Delta_3\omega_{\sigma_3k}^{3}=m\overset{0}{\omega}_k,\\
\Delta_0\overset{0}{\omega}_k-\Delta_1\omega_{\sigma_1k}^{01}-\Delta_2\omega_{\sigma_2k}^{02}-\Delta_3\omega_{\sigma_3k}^{03}=m\omega_k^0,\\
\Delta_1\overset{0}{\omega}_k-\Delta_0\omega_{\sigma_0k}^{01}-\Delta_2\omega_{\sigma_2k}^{12}-\Delta_3\omega_{\sigma_3k}^{13}=m\omega_k^1,\\
\Delta_2\overset{0}{\omega}_k-\Delta_0\omega_{\sigma_0k}^{02}+\Delta_1\omega_{\sigma_1k}^{12}-\Delta_3\omega_{\sigma_3k}^{23}=m\omega_k^2,\\
\Delta_3\overset{0}{\omega}_k-\Delta_0\omega_{\sigma_0k}^{03}+\Delta_1\omega_{\sigma_1k}^{13}+\Delta_2\omega_{\sigma_2k}^{23}=m\omega_k^3,\\
\Delta_0\omega_k^1-\Delta_1\omega_k^0-\Delta_2\omega_{\sigma_2k}^{012}-\Delta_3\omega_{\sigma_3k}^{013}=m\omega_k^{01},\\
\Delta_0\omega_k^2-\Delta_2\omega_k^0+\Delta_1\omega_{\sigma_1k}^{012}-\Delta_3\omega_{\sigma_3k}^{023}=m\omega_k^{02},\\
\Delta_0\omega_k^3-\Delta_3\omega_k^0+\Delta_1\omega_{\sigma_1k}^{013}+\Delta_2\omega_{\sigma_2k}^{023}=m\omega_k^{03},\\
\Delta_1\omega_k^2-\Delta_2\omega_k^1+\Delta_0\omega_{\sigma_0k}^{012}-\Delta_3\omega_{\sigma_3k}^{123}=m\omega_k^{12},\\
\Delta_1\omega_k^3-\Delta_3\omega_k^1+\Delta_0\omega_{\sigma_0k}^{013}+\Delta_2\omega_{\sigma_2k}^{123}=m\omega_k^{13},\\
\Delta_2\omega_k^3-\Delta_3\omega_k^2+\Delta_0\omega_{\sigma_0k}^{023}-\Delta_1\omega_{\sigma_1k}^{123}=m\omega_k^{23},\\
\Delta_0\omega_k^{12}-\Delta_1\omega_k^{02}+\Delta_2\omega_k^{01}+\Delta_3\overset{4}{\omega}_{\sigma_3k}=m\omega_k^{012},\\
\Delta_0\omega_k^{13}-\Delta_1\omega_k^{03}+\Delta_3\omega_k^{01}-\Delta_2\overset{4}{\omega}_{\sigma_2k}=m\omega_k^{013},\\
\Delta_0\omega_k^{23}-\Delta_2\omega_k^{03}+\Delta_3\omega_k^{02}+\Delta_1\overset{4}{\omega}_{\sigma_1k}=m\omega_k^{023},\\
\Delta_1\omega_k^{23}-\Delta_2\omega_k^{13}+\Delta_3\omega_k^{12}+\Delta_0\overset{4}{\omega}_{\sigma_0k}=m\omega_k^{123},\\
\Delta_0\omega_k^{123}-\Delta_1\omega_k^{023}+\Delta_2\omega_k^{013}-\Delta_3\omega_k^{012}=m\overset{4}{\omega}_k.
\end{align*}

The inner product (\ref{3.25}) in $K^r(4)$ can be extended to an inner product of inhomogeneous forms   by the rule
\begin{equation}\label{4.5}
(\Omega, \ \Phi)_V=\sum_{r=0}^4(\overset{r}{\omega}, \ \overset{r}{\varphi})_V,
\end{equation}
where $\Omega$,  $\Phi$  are given by (\ref{4.1}).
Given the inner product of inhomogeneous forms, we can consider the Dirac type operator
\begin{equation}\label{4.6}
D^c_{\pm}\equiv d^c\pm\delta^c: K(4)\rightarrow K(4)
\end{equation}
and its formally adjoint. Note that the operator (\ref{4.6}) does not respect the degree of forms because  $d^c: K^r(4) \rightarrow K^{r+1}(4)$ and  $\delta^c: K^{r+1}(4) \rightarrow K^r(4)$. Thus it is not possible to define  $D^c_{\pm}$ on $K^r(4)$, i.e., on homogeneous $r$-forms.

\begin{prop}
  The operator $D^c_{+}$ is self adjoint and the operator $D^c_{-}$ is anti self adjoint  with respect to the inner product (\ref{4.5}),
i.e.,
\begin{equation}\label{4.7}
(D^c_{\pm}\Omega, \ \Phi)_V=\pm(\Omega, \ D^c_{\pm}\Phi)_V.
\end{equation}
\end{prop}
\begin{proof}
We have
\begin{align*}\label{}
(D^c_{\pm}\Omega, \ \Phi)_V&=(d^c\overset{0}{\omega}, \ \overset{1}{\varphi})_V +(d^c\overset{1}{\omega}, \ \overset{2}{\varphi})_V+
(d^c\overset{2}{\omega}, \ \overset{3}{\varphi})_V+(d^c\overset{3}{\omega}, \ \overset{4}{\varphi})_V\\
&\pm(\delta^c\overset{1}{\omega}, \ \overset{0}{\varphi})_V\pm(\delta^c\overset{2}{\omega}, \ \overset{1}{\varphi})_V
\pm(\delta^c\overset{3}{\omega}, \ \overset{2}{\varphi})_V\pm(\delta^c\overset{4}{\omega}, \ \overset{3}{\varphi})_V.
\end{align*}
By (\ref{3.26}) this  implies (\ref{4.7}) immediately.
 \end{proof}

\section{Decomposition of the discrete Dirac-K\"{a}hler equation}
Let us introduce a discrete inhomogeneous form of type
\begin{equation}\label{5.1}
\overset{r\tau r}{\Omega}=\overset{r}{\omega}+\overset{\tau r}{\omega},
\end{equation}
where $\overset{r}{\omega}$ is an $r$-form of type (\ref{3.5})--(\ref{3.7}) and $\tau$ is given by (\ref{3.1}).
By analogy with the continual case \cite{Cantoni}  we define a discrete analog of the Duffin type equation by the rule
\begin{equation}\label{5.2}
(d^c-\delta^c)\overset{r\tau r}{\Omega}=m\overset{r\tau r}{\Omega}.
\end{equation}
This equation is equivalent to the following two equations
\begin{align}\label{5.3} \nonumber
d^c\overset{r}{\omega}=m\overset{\tau r}{\omega},\\
-\delta^c\overset{\tau r}{\omega}=m\overset{r}{\omega}.
\end{align}
For $r=0$ we have a discrete analog of the scalar Duffin equation and for $r=1$ we have a discrete analog of the vector Duffin equation (or the Proca equation). Similarly, if  $r=2$ and $r=3$ we obtain discrete analogs of the pseudovector and pseudoscalar Duffin equations. On the Duffin-Kemmer-Petiau formulation of  Klein-Gordon and field equations see \cite{Kanatchikov, Kruglov2} and references therein.
\begin{thm}
If the inhomogeneous form (\ref{5.1}) is a solution of the discrete Duffin equation (\ref{5.2}), then the $r$-form $\overset{r}{\omega}$ satisfies the discrete
 Klein-Gordon  equation (\ref{3.35}) for any $r=0,1,2,3$.
\end{thm}
\begin{proof}
Let $\overset{r\tau r}{\Omega}$ is a solution of (\ref{5.2}).
From  (\ref{5.3}) we obtain
\begin{equation*}
\delta^cd^c\overset{r}{\omega}=m\delta^c\overset{\tau r}{\omega}=-m^2\overset{r}{\omega}.
\end{equation*}
Since $\delta^c\delta^c=0$, from the second equation of (\ref{5.3}) we have  $\delta^c\overset{r}{\omega}=0$. Hence it follows that
\begin{equation*}
\delta^cd^c\overset{r}{\omega}+d^c\delta^c\overset{r}{\omega}=-m^2\overset{r}{\omega}.
\end{equation*}
 \end{proof}

 Set
 \begin{equation}\label{5.4}
\overset{r}{\omega}=\overset{r}{\omega}_1+\overset{r}{\omega}_2,
\end{equation}
where
\begin{equation*}
\overset{r}{\omega}_1=\overset{r}{\omega}+\frac{1}{m}\delta^c\overset{\tau r}{\omega} \quad \mbox{and} \quad
  \overset{r}{\omega}_2=-\frac{1}{m}\delta^c\overset{\tau r}{\omega}
\end{equation*}
for any $r=1,2,3$. Let us construct the forms
\begin{equation}\label{5.5}
\overset{01}{\Omega}=\overset{0}{\omega}+\overset{1}{\omega}_1, \quad  \overset{12}{\Omega}=\overset{1}{\omega}_2+\overset{2}{\omega}_1,
\quad  \overset{23}{\Omega}=\overset{2}{\omega}_2+\overset{3}{\omega}_1, \quad  \overset{34}{\Omega}=\overset{3}{\omega}_2+\overset{4}{\omega}.
\end{equation}
According to (\ref{5.4}) it turns out that
\begin{equation}\label{5.6}
\overset{01}{\Omega}+\overset{12}{\Omega}+\overset{23}{\Omega}+\overset{34}{\Omega}=\sum_{r=0}^4\overset{r}{\omega}=\Omega.
\end{equation}
\begin{thm}
Let  $\overset{r\tau r}{\Omega}$ is given by (\ref{5.5}). Then  $\overset{r\tau r}{\Omega}$ satisfies the discrete Duffin equation (\ref{5.2}) for   $r=0,1,2,3$, and any solution  $\Omega$ of the discrete  Dirac-K\"{a}hler equation (\ref{4.3}) can be uniquely represent as (\ref{5.6}).
\end{thm}
\begin{proof}
Since the differential operators $d$, $\delta$ and their discrete counterparts $d^c$, $\delta^c$ have the same properties, i.e., $d^cd^c=0$ and $\delta^c\delta^c=0$, the proof coincides with one in \cite{Cantoni}.
 \end{proof}
 \section{Massless limits  of the discrete Dirac-K\"{a}hler equation}

 We are interested in the massless discrete Dirac-K\"{a}hler equation.  Let us suppose  that the mass parameter  $m$ is equal to zero in equations (\ref{4.3}) and (\ref{4.4}). Consider the  transformation
 \begin{equation}\label{6.1}
\Omega\rightarrow\Omega+d^c\Phi,
\end{equation}
where $\Omega$ and  $\Phi$ are inhomogeneous  forms in view  (\ref{4.1}). The transformation (\ref{6.1}) is equivalent to
\begin{equation}\label{6.2}
\overset{0}{\omega}\rightarrow\overset{0}{\omega}, \qquad  \overset{r}{\omega}\rightarrow\overset{r}{\omega}+d^c\overset{\sigma r}\varphi,
\end{equation}
where $r=1,2,3,4$ and the shift operator $\sigma$ is given by (\ref{3.1}).

It should be noted that the transformation
\begin{equation*}
\alpha\rightarrow\alpha+d\varphi
\end{equation*}
is called in the continuum electromagnetic theory the gauge transformation. Here $\alpha$ is a differential  1-form associated with the  vector potential  and
$\varphi$ is a scalar function (gauge function).  The electromagnetic field is invariant under this gauge transformation.
By analogy, the transformation (\ref{6.2}) of discrete forms can be called gauge.
\begin{prop}
The discrete wave equation
\begin{equation}\label{6.3}
\Delta^c\Omega=0
\end{equation}
is invariant under the gauge transformation (\ref{6.1}), where the gauge form  $\overset{r}\varphi$ satisfies the condition
$\Delta^c\overset{r}\varphi=0$ for  $r=0,1,2,3$.
\end{prop}
\begin{proof}
Since  $\delta^cd^c\overset{0}\varphi=0$ and
$\delta^cd^c\overset{r}\varphi=-d^c\delta^c\overset{r}\varphi$ \ for $r=1,2,3$,
we have
\begin{align*}
\Delta^c(\overset{r}{\omega}+d^c\overset{\sigma r}\varphi)&=-(d^c\delta^c+\delta^cd^c)(\overset{r}{\omega}+d^c\overset{\sigma r}\varphi)=
\Delta^c\overset{r}{\omega}-d^c\delta^cd^c\overset{\sigma r}\varphi\\
&=\Delta^c\overset{r}{\omega}+\delta^cd^cd^c\overset{\sigma r}\varphi=\Delta^c\overset{r}{\omega}
 \end{align*}
 for any $r=1,2,3,4$.
 Hence
 \begin{equation*}
\Delta^c(\Omega+d^c\Phi)=\Delta^c\Omega.
\end{equation*}
\end{proof}
It is clear that Equation  (\ref{6.3}) is also invariant under the transformation
\begin{equation}\label{6.4}
\Omega\rightarrow\Omega+\delta^c\Phi,
\end{equation}
where the homogeneous components of $\Phi$ satisfy the condition $\Delta^c\overset{r}{\varphi}=0$ for $r=1,2,3,4$.
\begin{prop}
The discrete massless Dirac-K\"{a}hler equation
\begin{equation}\label{6.5}
(d^c-\delta^c)\Omega=0
\end{equation}
is invariant under the following transformation
\begin{equation}\label{6.6}
\Omega\rightarrow\Omega+d^c\Phi-\delta^c\Phi,
\end{equation}
where $\Phi$  satisfies equation (\ref{6.3}).
\end{prop}
\begin{proof}
 Since, by assumption $\Delta^c\Phi=0$ and by (\ref{3.33}),  one has
 \begin{equation*}\label{}
(d^c-\delta^c)(\Omega+d^c\Phi-\delta^c\Phi)=(d^c-\delta^c)\Omega+(d^c-\delta^c)^2\Phi=(d^c-\delta^c)\Omega.
\end{equation*}
\end{proof}
Note that the transformation (\ref{6.6})  can be  written  more explicitly  as
\begin{equation*}\label{}
\overset{0}{\omega}\rightarrow\overset{0}{\omega}-\delta^c\overset{1}\varphi, \qquad
\overset{4}{\omega}\rightarrow\overset{4}{\omega}+d^c\overset{3}\varphi, \qquad
\overset{r}{\omega}\rightarrow\overset{r}{\omega}+d^c\overset{\sigma r}\varphi-\delta^c\overset{\tau r}\varphi,
\end{equation*}
where  $r=1,2,3$.

It is known (see \cite{Kruglov1, Pletyuhov}) that in the classical (continuous) theory there are three massless analogs of Equation (\ref{4.3}). Now we describe discrete
counterparts of these cases. Let introduce two mass parameters $m_1, m_2$ and  put $m=\sqrt{m_1m_2}$. Then the system  (\ref{4.4})  with two mass parameters can be written as
\begin{align}\label{6.7}
-\delta^c\overset{1}{\omega}=m_1\overset{0}{\omega}, \nonumber \\
d^c\overset{0}{\omega}-\delta^c\overset{2}{\omega}=m_2\overset{1}{\omega}, \nonumber \\
d^c\overset{1}{\omega}-\delta^c\overset{3}{\omega}=m_1\overset{2}{\omega},\\
d^c\overset{2}{\omega}-\delta^c\overset{4}{\omega}=m_2\overset{3}{\omega}, \nonumber \\
d^c\overset{3}{\omega}=m_1\overset{4}{\omega}.  \nonumber
\end{align}
It is easy to check that if the collection of all  forms $\overset{r}{\omega}$ is a solution of (\ref{6.7}), then the inhomogeneous form $\Omega$ (\ref{4.1})  satisfies the equation
\begin{equation}\label{6.8}
\Delta^c\Omega=m_1m_2\Omega.
\end{equation}
Hence, if any the mass parameters $m_1, m_2$ (or both simultaneously) is equal to zero then
the equation (\ref{6.8}) corresponds to the massless case. Of course, in the case $m_1=m_2=0$ the system (\ref{6.7}) is equivalent to equation (\ref{6.5}).
Let us consider the case  $m_2=0$ and $m_1\neq0$. In the continuous theory this is the so-called electromagnetic massless limit of the Dirac-K\"{a}hler equation \cite{Kruglov1, Pletyuhov}.
A discrete analog for this case has the form
\begin{align}\label{6.9}
-\delta^c\overset{1}{\omega}=m_1\overset{0}{\omega}, \nonumber \\
d^c\overset{0}{\omega}-\delta^c\overset{2}{\omega}=0, \nonumber \\
d^c\overset{1}{\omega}-\delta^c\overset{3}{\omega}=m_1\overset{2}{\omega},\\
d^c\overset{2}{\omega}-\delta^c\overset{4}{\omega}=0, \nonumber \\
d^c\overset{3}{\omega}=m_1\overset{4}{\omega}.  \nonumber
\end{align}
\begin{prop}
The system (\ref{6.9})
is invariant under the following transformation
\begin{equation}\label{6.10}
\overset{1}{\omega}\rightarrow\overset{1}{\omega}+d^c\overset{0}{\varphi}, \qquad \overset{3}{\omega}\rightarrow\overset{3}{\omega}+\delta^c\overset{4}{\varphi},
\end{equation}
where the gauge forms $\overset{0}{\varphi}$ and $\overset{4}{\varphi}$  satisfy the  equations $\Delta^c\overset{0}{\varphi}=0$ and $\Delta^c\overset{4}{\varphi}=0$.
\end{prop}
\begin{proof}
Under the gauge transformation (\ref{6.10}), the left hand side of the corresponding equations of (\ref{6.9}) becomes
 \begin{equation*}\label{}
-\delta^c(\overset{1}{\omega}+d^c\overset{0}{\varphi})=-\delta^c\overset{1}{\omega}-\delta^cd^c\overset{0}{\varphi}=
-\delta^c\overset{1}{\omega}+\Delta^c\overset{0}{\varphi}=-\delta^c\overset{1}{\omega},
\end{equation*}
\begin{equation*}\label{}
d^c(\overset{1}{\omega}+d^c\overset{0}{\varphi})-\delta^c(\overset{3}{\omega}+\delta^c\overset{4}{\varphi})=d^c\overset{1}{\omega}+d^cd^c\overset{0}{\varphi}
-\delta^c\overset{3}{\omega}-\delta^c\delta^c\overset{4}{\varphi}=d^c\overset{1}{\omega}-\delta^c\overset{3}{\omega}
\end{equation*}
and
\begin{equation*}\label{}
d^c(\overset{3}{\omega}+\delta^c\overset{4}{\varphi})=d^c\overset{3}{\omega}+d^c\delta^c\overset{4}{\varphi}=
d^c\overset{3}{\omega}-\Delta^c\overset{4}{\varphi}=d^c\overset{3}{\omega}.
\end{equation*}
\end{proof}

In the case  $m_1=0$ and $m_2\neq0$ we obtain
\begin{align}\label{6.11}
-\delta^c\overset{1}{\omega}=0, \nonumber \\
d^c\overset{0}{\omega}-\delta^c\overset{2}{\omega}=m_2\overset{1}{\omega}, \nonumber \\
d^c\overset{1}{\omega}-\delta^c\overset{3}{\omega}=0,\\
d^c\overset{2}{\omega}-\delta^c\overset{4}{\omega}=m_2\overset{3}{\omega}, \nonumber \\
d^c\overset{3}{\omega}=0.  \nonumber
\end{align}
This system is a discrete analog of the Kalb-Ramond field equations (notoph equations) \cite{Malik, Pletyuhov}.

\begin{prop}
The system (\ref{6.11})
is invariant under the following transformation
\begin{equation}\label{6.12}
\overset{0}{\omega}\rightarrow\overset{0}{\omega}-\delta^c\overset{1}{\varphi}, \qquad \overset{2}{\omega}\rightarrow\overset{2}{\omega}+d^c\overset{1}{\varphi}-\delta^c\overset{3}{\varphi}, \qquad \overset{4}{\omega}\rightarrow\overset{4}{\omega}+d^c\overset{3}{\varphi},
\end{equation}
where the gauge forms $\overset{1}{\varphi}$ and $\overset{3}{\varphi}$  satisfy the  equations $\Delta^c\overset{1}{\varphi}=0$ and $\Delta^c\overset{3}{\varphi}=0$.
\end{prop}
\begin{proof}
Substituting (\ref{6.12}) into the left hand side of the second and fourth equations of (\ref{6.11}) we obtain
 \begin{align*}\label{}
d^c(\overset{0}{\omega}-\delta^c\overset{1}{\varphi})-\delta^c(\overset{2}{\omega}+d^c\overset{1}{\varphi}-\delta^c\overset{3}{\varphi})&=
d^c\overset{0}{\omega}-d^c\delta^c\overset{1}{\varphi}
-\delta^c\overset{2}{\omega}-\delta^cd^c\overset{1}{\varphi}-\delta^c\delta^c\overset{3}{\varphi}\\
&=d^c\overset{0}{\omega}-\delta^c\overset{2}{\omega}+\Delta^c\overset{1}{\varphi}=d^c\overset{0}{\omega}-\delta^c\overset{2}{\omega}
\end{align*}
and
\begin{align*}\label{}
d^c(\overset{2}{\omega}+d^c\overset{1}{\varphi}-\delta^c\overset{3}{\varphi})-\delta^c(\overset{4}{\omega}+d^c\overset{3}{\varphi})&=
d^c\overset{2}{\omega}+d^cd^c\overset{1}{\varphi}-\delta^cd^c\overset{3}{\varphi}
-\delta^c\overset{4}{\omega}-\delta^cd^c\overset{3}{\varphi}\\
&=d^c\overset{2}{\omega}-\delta^c\overset{4}{\omega}+\Delta^c\overset{3}{\varphi}=d^c\overset{2}{\omega}-\delta^c\overset{4}{\omega}.
\end{align*}
\end{proof}

We should remark that in the continual case the two massless limits are related to self-dual and anti-self-dual Dirac-K\"{a}hler fields. The Dirac-K\"{a}hler
equation decomposes on its self-dual and anti-self-dual parts by using the projections operators $P_\pm=\frac{1}{2}(I\pm\ast)$, where $\ast$ is the Hodge star operator.
The reason that we do not discuss this relation here is that in our formalism the operation $\ast^2$ is equivalent to a shift with corresponding sign. This is somewhat different from the continual case, where the operation  $\ast$ is either an involution or antiinvolution, i.e,  $\ast^2=\pm I$. Therefore there are difficulties in getting discrete counterparts of the projections operators.

\section{Conclusion}
Within a differential geometric discretisation approach  a new discrete analog of the Dirac-K\"{a}hler equation has been proposed.
In this discrete model some key aspects  of Hodge theory relevant to physics are captured. We have considered three massless limits of  the discrete Dirac-K\"{a}hler
equation. The intrinsic notation for the discrete equations, given in terms of the difference operators $d^c$ and $\delta^c$, allows us to discus the intrinsic form of gauge transformations and the invariance of the  discrete massless system under such transformations has been studied.

The  Dirac-K\"{a}hler formalism can be formulated by the Clifford product of differential forms. It is important to clarify the relation between the Dirac-K\"{a}hler and Dirac  equation. The identification of inhomogeneous differential forms and the gamma (Dirac) matrices is possible due to the Clifford product.
It would be interesting to introduce a Clifford product  which acts on the space of our discrete inhomogeneous forms. Giving this we should generalize the discrete Dirac operator introduced above. We also expect to give a better explanation  to transformations properties of a discrete model. These directions must be investigated and we hope to treat them further in future work.

\medskip

{\bf Acknowledgement}
\medskip

The author is grateful to the referee for bringing \cite{Catterall, Rabin} to his attention.

\end{document}